\RequirePackage{fix-cm}
\documentclass[smallextended]{svjour3}       
\smartqed  
\usepackage{graphicx}
\usepackage{subfigure}
\usepackage{epstopdf}
\usepackage{graphics} 
\usepackage{amsmath,amsfonts,amssymb,dsfont}
\usepackage{url}

\def\R{\mathbb{R}}
\def\Nor{\mathcal{N}}
\def\I{\mathrm{I}}

\def\m{\mathrm{m}}

\def\F{\mathcal{F}}
\def\G{\mathcal{G}}

\def\M{\mathrm{M}}

\def\diag{\mathrm{diag}}

\def\v{\mathrm{v}}

\def\1{\mathds{1}}

\def\for{\mbox{  for }}

\def\det{\mathrm{det}}
\def\tr{\mathrm{tr}}

\usepackage{latexsym}
\begin{document}

\title{Optimal Rescaling and the Mahalanobis Distance}
\subtitle{}


\author{Przemys³aw Spurek    \and
        Jacek Tabor 
}


\institute{P. Spurek \at
              Faculty of Mathematics and Computer Science, Jagiellonian University\\ 
              \L ojasiewicza 6\\
              30-348 Krak\'ow\\ 
             Poland\\
              \email{przemyslaw.spurek@ii.uj.edu.pl  }           
           \and
           J. Tabor \at
              Faculty of Mathematics and Computer Science, Jagiellonian University\\ 
              \L ojasiewicza 6\\
              30-348 Krak\'ow\\ 
             Poland\\
              \email{jacek.tabor@ii.uj.edu.pl  }        
}

\date{Received: date / Accepted: date}

\maketitle

\begin{abstract}
One of the basic problems in data analysis lies in choosing the optimal rescaling (change of coordinate system) to study properties of a given data-set $Y$. The classical Mahalanobis approach has its basis in the classical normalization/rescaling formula $Y \ni y \to \Sigma_Y^{-1/2} \cdot (y-\m_Y)$,
where $\m_Y$ denotes the mean of $Y$ and $\Sigma_Y$ the covariance matrix .

Based on the cross-entropy we generalize this approach
and define the parameter which measures the fit of a 
given affine rescaling of $Y$ compared to the Mahalanobis one.
This allows in particular to find an optimal change of coordinate system which satisfies some additional conditions. In particular we show that in the case when  we put origin of coordinate system in $ \m $ the optimal choice is given by the transformation $Y \ni y \to \Sigma_Y^{-1/2} \cdot (y-\m_Y)$, where
$$
\Sigma=\Sigma_Y(\Sigma_Y-\frac{(\m-\m_Y)(\m-\m_Y)^T}{1+\|\m-\m_Y\|_{\Sigma_Y}^2} )^{-1}\Sigma_Y.
$$
\end{abstract}

\section{Introduction}


One of the crucial problem in statistics, compression, discrimination analysis
and in general in data mining is how to choose the optimal linear coordinate system and define distance which ``optimally'' underlines the internal structure of the data \cite{Bo-Gr,DM-JR,Ha-Ka,Ja-Sa,Ra-Ma,Re,Ti,Kr-Ka,ma_dis_a_3,ma_dis_a_4,ma_dis_a_5}.
The typical answer is the Mahalanobis distance \cite{ma_dis_a_1,ma_dis_a_2,Ma,DM-JR}, which is strictly connected with PCA \cite{jolliffe2005principal,pca_a_1,pca_a_2}. 
In some other cases we choose just a weighted Euclidean distance.
Our approach uses a method based on MLE (Maximum Likelihood Estimation) \cite{Le-Ca,Bo,mle_a_3}.

To explain it more precisely consider the data-set $Y \subset \R^N$.
If we allow the translation of the origin of coordinate system, 
in one dimensional case we usually apply the normalization:
$s:Y \ni y \to \sigma_Y^{-1}(y-\m_Y)$, 
which in the multivariate case is replaced by
$$
s:Y \ni y \to \Sigma_Y^{-1/2}(y-\m_Y),
$$
where $\m_Y$ denotes the mean and $\Sigma_Y$ the covariance matrix of $Y$. Then we obtain 
that the coordinates are uncorrelated, and the covariance matrix equals to identity. Taking the distance between the transformation of points $x,y$:
$$
\|sx-sy\|^2=(sx-sy)^T (sx-sy)=(x-y)^T \Sigma_{Y}^{-1}(x-y)
$$
we arrive naturally at the definition of the Mahalanobis distance:
$$
\|x-y\|_{\Sigma}^2:=(x-y)^T \Sigma^{-1}(x-y).
$$
If we do not allow the translation of the origin away from zero, which in some cases is natural, we usually only scale/normalize each coordinate by dividing it by its mean. This approach usually has good results if the standard deviation is small in comparison to the mean. In the opposite case, when the mean is small, dividing by it may ``unnaturally'' widen the variable under consideration.

The main results of the paper provide the optimal change of coordinate system which satisfies some additional conditions. In particular we show that in the case when  we put origin of coordinate system in $ \m $ the optimal choice is given by the mapping $Y \ni y \to \Sigma_Y^{-1/2} \cdot (y-\m_Y)$, where
$$
\Sigma=\Sigma_Y(\Sigma_Y-\frac{(\m-\m_Y)(\m-\m_Y)^T}{1+\|\m-\m_Y\|_{\Sigma_Y}^2} )^{-1}\Sigma_Y.
$$

\begin{figure}[hh]
  \begin{center}
\subfigure[]
{\label{s1}
\begin{minipage}{0.3\linewidth}
\includegraphics[width=1.4in]{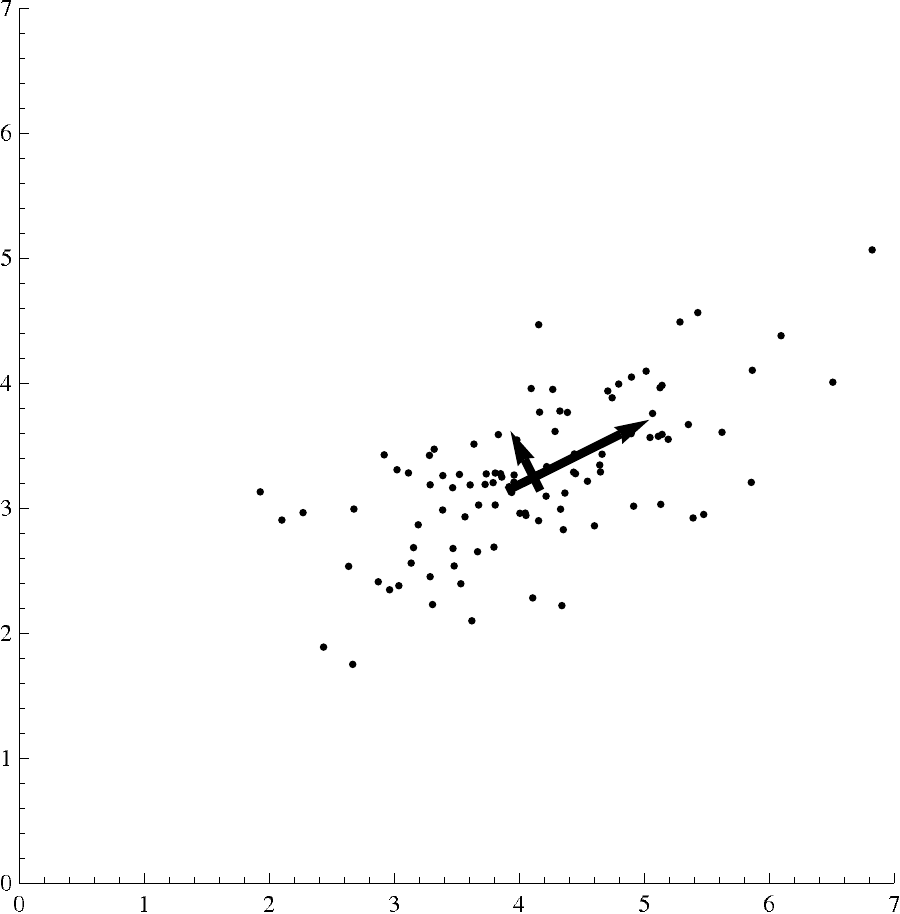}   

\includegraphics[width=1.4in]{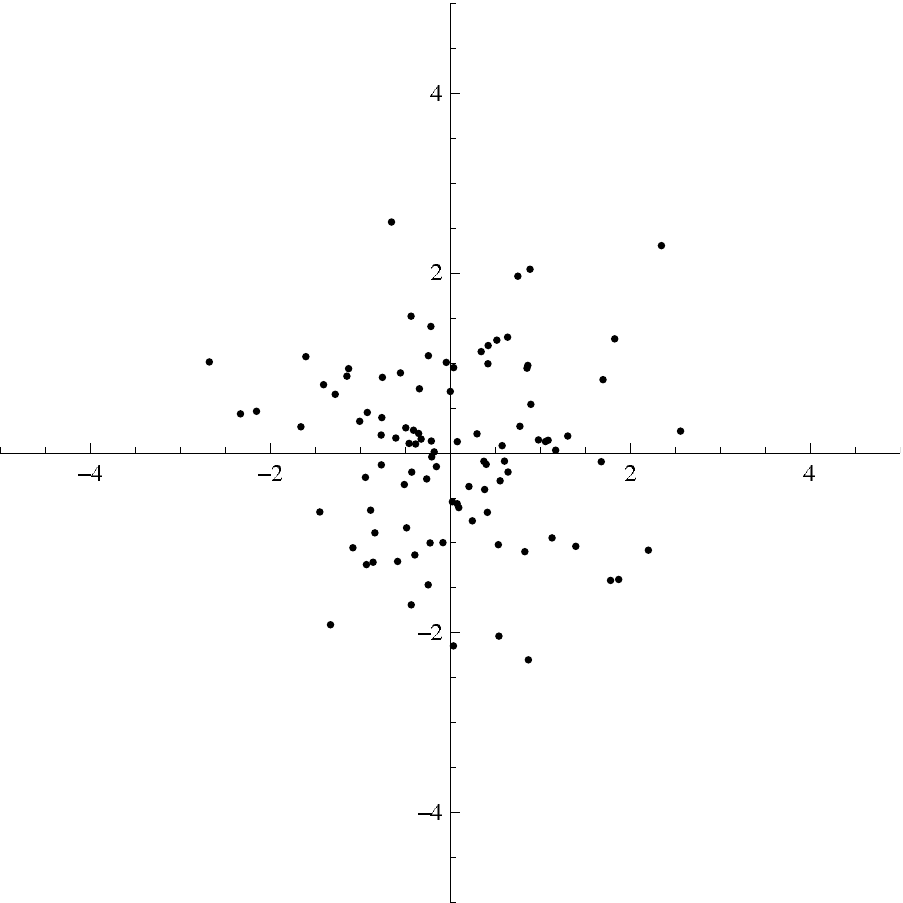} 
\end{minipage}}
\subfigure[] 
{\label{s2}
\begin{minipage}{0.3\linewidth}
\includegraphics[width=1.4in]{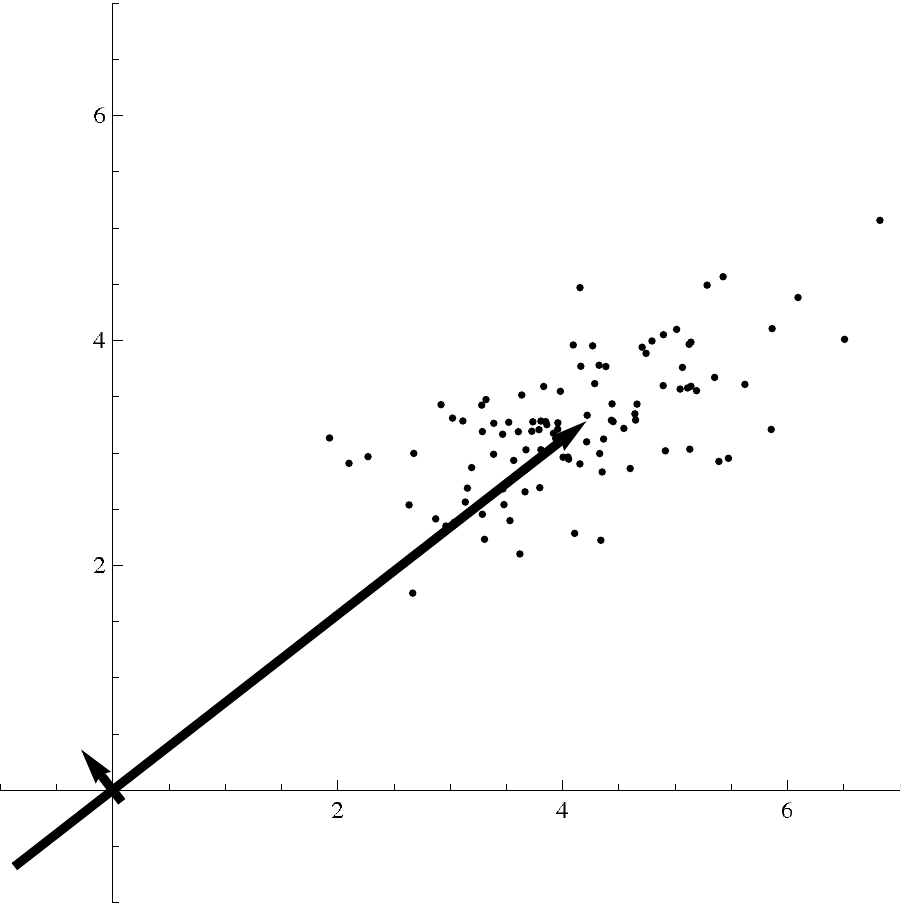}  

\includegraphics[width=1.4in]{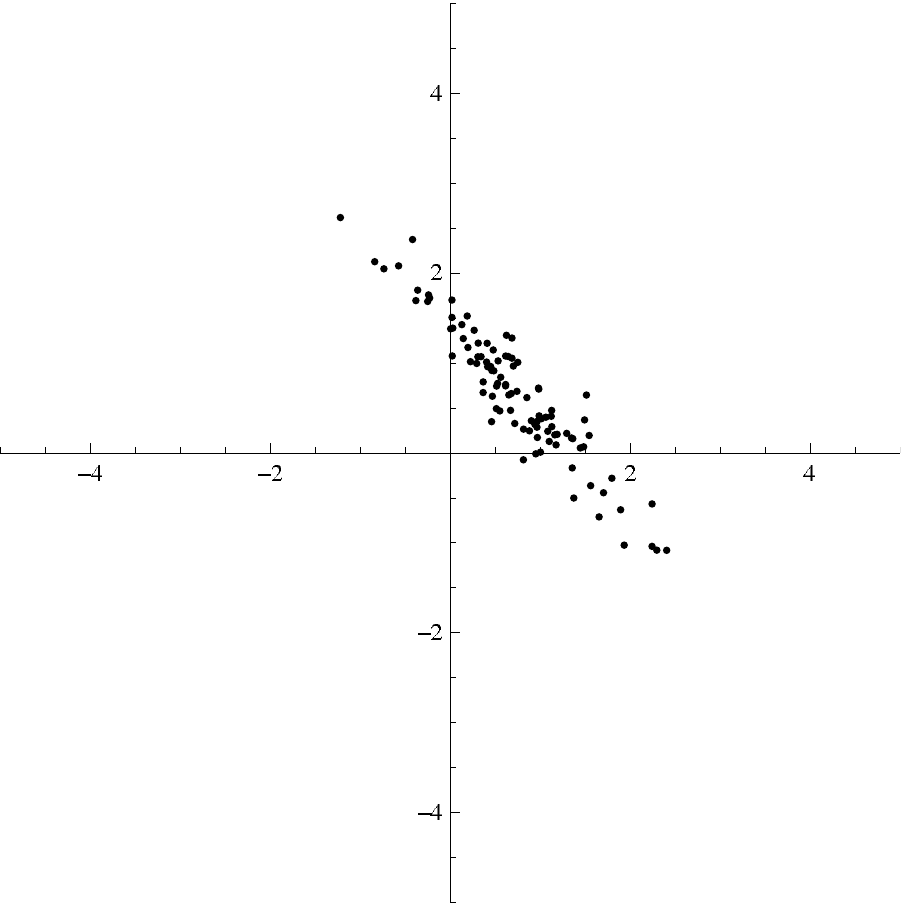} 
\end{minipage}
}
\subfigure[]
{\label{s3}
\begin{minipage}{0.3\linewidth}
\includegraphics[width=1.4in]{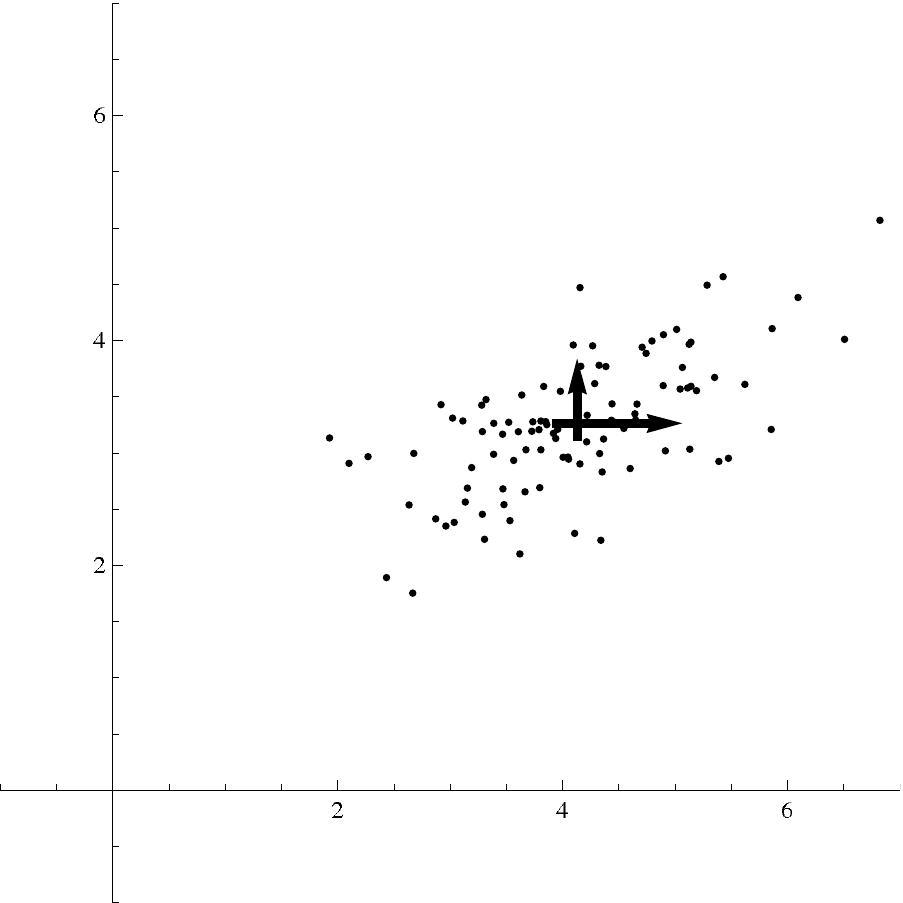} 

\includegraphics[width=1.4in]{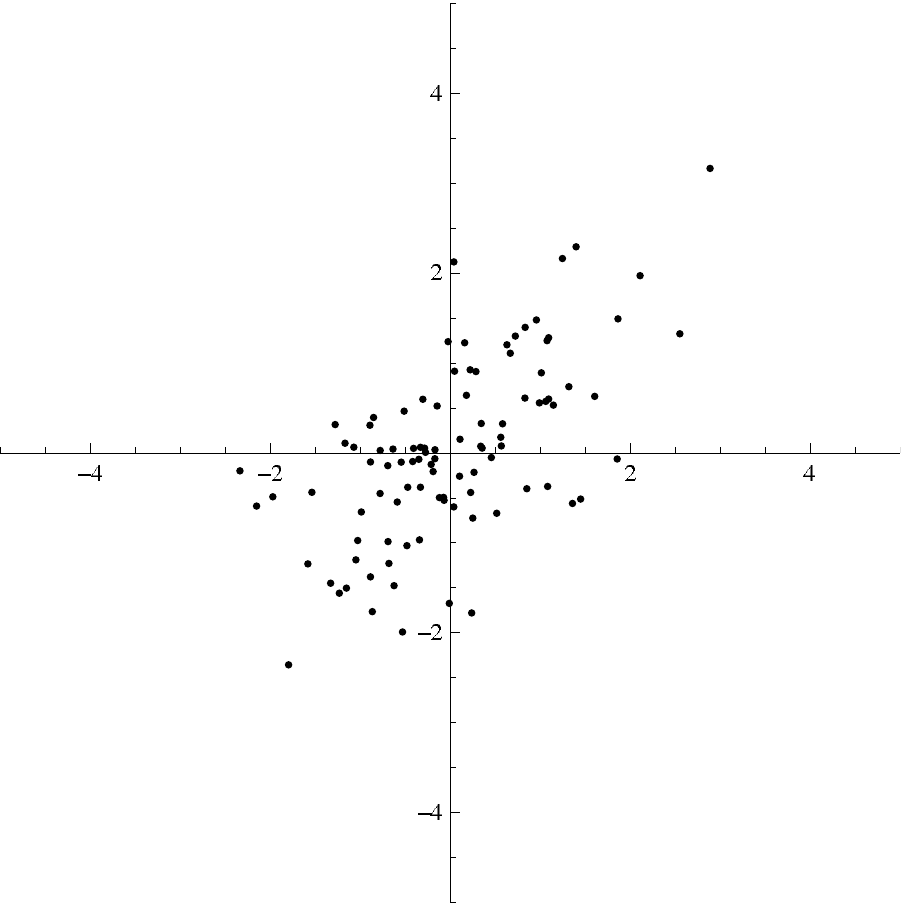} 
\end{minipage}
}
  \end{center}
  \caption{Optimal coordinate systems and the data rewritten in new bases in various situations:  a) optimal  Mahalanobis base; b) optimal base in the center in $[0,0]$ c) optimal base in the case of optimal rescaling of each coordinates; }
  \label{fig:cs_2}
\end{figure}

\begin{example}\label{ex:1}
We illustrate our results on a two-dimensional data drawn from the normal distribution
with mean $\m=[3,4]^T$ and covariance $\Sigma=\left[\begin{array}{cc}1&0.3\\0.3&0.6\end{array}\right]$. In the Figure \ref{s1} we present the base which represents the best (Mahalanobis) rescaling and the data in new coordinates.  In \ref{s2} we present optimal base in the case we do not allow to move the origin of the coordinate system from zero and the data in new base. In the Figure \ref{s3} we present optimal rescaling in the case when we the change of the origin, but we restrict to rescale the coordinates separately and the data in new base. 
\end{example}

To explain more precisely what we mean by optimal coordinates (in the 
given class of coordinate systems)
we need to introduce the function which
measures the ``match'' of a given coordinate system to the data. Suppose that we are given a base $\v=(v_1,\ldots,v_N)$ of $\R^N$ and we put an origin of coordinate system at $\m \in \R^N$. Then
by $\Nor_{[\m,\v]}$ we denote the ``normalized'' Gaussian density with respect to the base $\v$ with center at $\m$. In other words
$$
\Nor_{[\m,\v]}(m+x_1v_1+\ldots+x_N v_N)=
\frac{1}{(2\pi)^{N/2}}e^{-(x_1^2+\ldots+x_N^2)/2}.
$$
Observe that $\Nor_{[\m,\v]}=\Nor_{(\m,(vv^T)^{-1})}$,
where by $\Nor_{(\m,\Sigma)}$ we denote the normal density with mean $\m$
and covariance $\Sigma$. Then we can measure the ``match/fit'' of the coordinate system by the cross-entropy
\begin{equation} \label{e0}
H^{\times}(Y\|\Nor_{[\m,\v]})
\end{equation}
of $Y$ with respect to $\Nor_{[\m,\v]}$.
One can easily observe that the minimum in \eqref{e0} is attained by the Mahalanobis coordinate system, where the mean and covariance is that of $Y$, and thus we obtain a confirmation that the commonly used procedure is reasonable. 
This allows to define the match of data $Y$
with respect to the coordinate system (or more precisely the
corresponding Gaussian density) by the formula
\begin{equation} \label{eq:M_la}
\M(Y\|\Nor_{(\m,\Sigma)}):=H^{\times}(Y\|\Nor_{(\m,\Sigma)})-H^{\times}(Y\|\Nor_Y) \geq 0,
\end{equation}
where $\Nor_Y$ denotes the normal density with mean and covariance of $Y$.
From the well-known formulas \cite{Ni-No,Pe} concerning relative entropy and cross-entropy:
\begin{equation} \label{e1.5}
\begin{array}{l}
H^{\times}(Y\|\Nor_{(\m,\Sigma)})= \frac{N}{2} \ln(2\pi)+\frac{1}{2}\|\m-\m_Y\|^2_{\Sigma}+\frac{1}{2}\tr(\Sigma^{-1}\Sigma_Y)+\frac{1}{2}\ln \det \Sigma,\\[1ex]
H^{\times}(Y\|\Nor_Y) =\frac{N}{2} \ln(2\pi e)+\frac{1}{2}\ln \det \Sigma_Y,
\end{array}
\end{equation}
we can easily deduce the formula for our match function for the Gaussian distribution $\Nor(\m,\Sigma)$ in $\R^N$:
\begin{equation} \label{e2}
\M(Y\|\Nor_{(\m,\Sigma)})= \frac{1}{2}\left(\|\m-\m_Y\|^2_{\Sigma}+\tr(\Sigma^{-1}\Sigma_Y)-\ln \det(\Sigma^{-1}\Sigma_Y)-N\right).
\end{equation}
We can now 
consider coordinate systems identified with
respective subclasses of Gaussian distributions. Consequently for a family
$\F$ of Gaussian distributions we put
$$
\M(Y\|\F):=\inf_{f \in \F}\M(Y\|f) = \inf_{f \in \F} H^{\times}(Y\| f)-H^{\times}(Y\|\Nor_Y).
$$
Clearly, to find $f$ which minimizes the above it is equivalent to find
$f$ which minimizes the cross-entropy
$
H^{\times}(Y\|\F) := \inf_{f \in \F} H^{\times}(Y\|f)
$.
Since cross-entropy and log-likelihood functions differ only by sign, we arrive
at the typical MLE \cite{Le-Ca, Bo,mle_a_3} problem. 

As $\F$ we consider the following typical subfamilies of all Gaussians $\G$:
\newline
\begin{tabular}{ l c p{9.5cm}} 
$\G_{\m}$ & -- & Gaussian densities with mean at $\m$;
\\[0.5ex] 
$\G_{s\I}$ & -- & Gaussian densities with covariance proportional to identity; \\[0.5ex] 
$\G_{\m,s\I}$ & -- & Gaussian densities with mean at $\m$ and covariance proportional to identity; \\[0.5ex] 
$\G_{\diag}$ & -- & Gaussian densities with diagonal covariance; \\[0.5ex]
$\G_{\m,\diag}$ & -- & Gaussian densities with mean at $\m$ and diagonal covariance. \\[0.5ex]
\end{tabular}
\newline
Observe for example that the use of family $\G_{\m,\diag}$ means
that we consider the coordinate systems which have origin at $\m$ 
and which have axes parallel to the original (cartesian) ones.
Moreover, if we have found the optimal Gaussian $\Nor_{(\m,\Sigma)}$
in our class of densities (identified with coordinate systems), we can transform the data $Y$ into those coordinates by the affine 
transformation $Y \ni y \to \Sigma^{-1/2}(y-\m) \in \R^N$.

\section{Rescaling}

Assume that we have fixed the origin of the coordinate system at the point $\m$
and that we want to find how we should (uniformly) rescale the coordinates to optimally fit the data. This means that we search for $s$ such that 
$s \to \M(Y\|\Nor_{(\m_Y,s\I)})$ attains minimum. Since 
\begin{equation} \label{e3}
H^{\times}(Y\|\Nor_{(\m_Y,s\I)})=\frac{1}{2}\big((\tr(\Sigma_Y)+\|\m-\m_Y\|^2)s^{-1}+N \ln s+N \ln(2\pi)\big),
\end{equation}
by the trivial calculations we obtain that the above function attains its minimum 
\begin{equation} \label{si}
\frac{N}{2}\big(\ln \frac{\tr (\Sigma_Y)+\|\m-\m_Y\|^2}{N}+\ln(2\pi e) \big)
\end{equation}
for 
\begin{equation} \label{e3.5}
s=\frac{\tr(\Sigma_Y)+\|\m-\m_Y\|^2}{N}.
\end{equation} 
Thus applying (\ref{e1.5}) to (\ref{eq:M_la}) we have arrived at the following theorem. 

\begin{theorem} \label{bas}
Let $Y$ be a data-set with invertible covariance matrix and $\m$ be fixed. Then $\M(Y\| \G_{\m,s \I})$ is minimized for $s=(\tr(\Sigma_Y)+\|\m-\m_Y\|^2)/N$, and equals 
$$
\M(Y\| \G_{\m,s \I}) = \frac{N}{2}\ln \frac{\tr(\Sigma_Y)+\|\m-\m_Y\|^2}{N}-\frac{1}{2}\ln \det \Sigma_Y.
$$
\end{theorem}



If we allow in above Theorem the change of the origin, to minimize the value of $\M$ we have to clearly put $\m$ at $\m_Y$:

\begin{corollary}
Let $Y$ be a data-set with invertible covariance matrix. Then $\M(Y\| \G_{s \I})$ is minimized for $\m=\m_Y$, $s=\frac{1}{N}\tr(\Sigma_Y)$, and equals 
$$
\M(Y\| \G_{s \I} )=\frac{N}{2}\ln (\tr (\Sigma_Y)/N )-\frac{1}{2}\ln \det \Sigma_Y.
$$
\end{corollary}

\begin{remark}
Assume that we want to move the origin to $\m$, and uniformly 
rescale. Then Theorem \ref{bas} implies that
$$
y \to (y-\m)/\sqrt{\frac{1}{N}(\tr(\Sigma_Y)+\|\m-\m_Y\|^2)}
$$
is the optimal rescaling of this type. Thus in the case of univariate data, the optimal
rescaling when we do not change the origin of the coordinate system 
is given by
$y \to y/\sqrt{\sigma_Y^2+\m_Y^2}=y/\sqrt{E(Y^2)}$.
\end{remark}

We consider the case when we allow to rescale each coordinate $Y_i$ of $Y=(Y_1,\ldots,Y_N)$ separately. Then the optimal change of 
coordinates of this form is given by optimal rescaling on each coordinate.

\begin{corollary}
Let $Y$ be a data-set with invertible covariance matrix and $\m$ be fixed. Then $\M(Y\| \G_{\m,\diag})$ is minimized for $s_i=(\Sigma_Y)_{ii}+|\m_i-(\m_Y)_i|^2$, and equals 
$$
\M(Y\| \G_{\m,\diag}) = \frac{1}{2}\sum_{i=1}^N \ln \big((\Sigma_Y)_{ii}+|\m_i-(\m_Y)_i|^2\big)-\frac{1}{2}\ln \det \Sigma_Y.
$$
\end{corollary}

\begin{proof}
For $\Nor_{(\m,\diag(s_1,\ldots,s_n))} \in \G_{\m,\diag}$ we have
$$
\begin{array}{l}
\M(Y\| \Nor_{(\m,\diag(s_1,\ldots,s_n))})=
H^{\times}(Y\|\Nor_{(\m,\diag(s_1,\ldots,s_n))})-H^{\times}(Y\|\Nor_Y) \\[1ex]
=\sum \limits_{i=1}^N H^{\times}(Y_i\|\Nor_{(\m_i,s_i)})-H^{\times}(Y\|\Nor_Y).
\end{array}
$$
By applying \eqref{si} to the univariate data $Y_i$ and \eqref{e1.5}
we obtain that the minimum of $\M(Y\|\Nor_{(\m,\diag(s_1,\ldots,s_n))})$
is realized for $s_i=(\Sigma_Y)_{ii}+|\m_i-(\m_Y)_i|^2)$ and equals
$$
\frac{1}{2}\sum_{i=1}^N\big(\ln ((\Sigma_Y)_{ii}+|\m_i-(\m_Y)_i|^2)+\ln(2\pi e) \big)-\frac{1}{2}\big(N \ln(2\pi e)+\ln \det \Sigma_Y\big)
$$ 
$$
=\frac{1}{2}\sum_{i=1}^N \ln ((\Sigma_Y)_{ii}+|\m_i-(\m_Y)_i|^2)-\frac{1}{2}\ln \det \Sigma_Y.
$$
\end{proof}

If we additionally allow the change of the origin, we should put $\m=\m_Y$ and the rescaling takes the form
$y \to (y-\m_Y)/\sqrt{\tr(\Sigma_Y)/N}$.



Now we focus our attention on the problem how to find the optimal coordinate system in the general case. To do so we first need to present
a simple consequence of the famous von Neuman trace inequality 
\cite{Gr,Mi}.

\medskip

\noindent{\bf Theorem [von Neumann trace inequality].}{\em \/ Let $E,F$ be complex $N \times N$ matrices. Then
\begin{equation} \label{neu}
|\tr(EF) | \leq \sum_{i=1}^N s_i(E)\cdot s_i(F),
\end{equation}
where $s_i(D)$ denote the ordered (decreasingly)
singular values of matrix $D$.}

\medskip

Let us recall that for the symmetric positive matrix its
eigenvalues coincide with singular values.
Given $\lambda_1,\ldots,\lambda_N \in \R$ by $S_{\lambda_1,\ldots,\lambda_N}$ we denote the set of all symmetric matrices with eigenvalues $\lambda_1,\ldots,\lambda_N$.

\begin{proposition}
Let $B$ be a symmetric nonnegative matrix with eigenvalues $\beta_1 \geq \ldots \geq\beta_N \geq 0$ .
Let $0 \leq \lambda_1 \leq \ldots \leq \lambda_N$ be fixed.
Then 
$$
\min_{A \in S_{\lambda_1,\ldots,\lambda_N}} \tr(AB)=\sum_i \lambda_i \beta_i.
$$ 
\end{proposition}

\begin{proof}
Let $e_i$ denote the orthogonal basis build from the eigenvectors of $B$, and let 
operator $\bar A$ be defined in this base by $\bar A(e_i)=\lambda_i e_i$. Then trivially 
$$
\min_{A \in S_{\lambda_1,\ldots,\lambda_N}} \tr(AB) \leq
\tr(\bar AB)=\sum_i \lambda_i \beta_i.
$$

To prove the inverse inequality we will use the
von Neumann trace inequality. Let $A \in S_{\lambda_1,\ldots,\lambda_N}$ be arbitrary. We apply the inequality \eqref{neu} for $E=\lambda_N \I-A$, $F=B$.
Since $E$ and $F$ are symmetric nonnegatively defined matrices, their eigenvalues
$\lambda_N-\lambda_i$ and $\beta_i$ coincide with singular values, and therefore by \eqref{neu}
\begin{equation} \label{nu2}
\tr((\lambda_N\I-A)B) \leq \sum_i(\lambda_N-\lambda_i)\beta_i=
\lambda_N \sum_i \beta_i -\sum_i \lambda_i \beta_i.
\end{equation}
Since 
$$
\tr((\lambda_N\I-A)B)=\lambda_N \sum_i \beta_i -\tr(AB),
$$
from inequality \eqref{nu2} we obtain that
$\tr(AB) \geq \sum_i \lambda_i \beta_i$.
\end{proof}

Now we proceed to the main result of the paper.

\begin{theorem}
Let $Y$ be a data-set and $\m \in \R^N$ be fixed. Then 
$$
\M(Y\|\G_{\m})=\frac{1}{2}\ln(1+\|\m-\m_Y\|^2_{\Sigma_Y}) 
$$
and is attained for the density $\Nor(\m,\Sigma) \in \G_{\m}$, where
$$
\Sigma=\Sigma_Y(\Sigma_Y-\frac{(\m-\m_Y)(\m-\m_Y)^T}{1+\|\m-\m_Y\|_{\Sigma_Y}^2} )^{-1}\Sigma_Y.
$$
\end{theorem}

\begin{proof}
Let us first observe that by applying substitution
$$
A=\Sigma_Y^{1/2}\Sigma^{-1}\Sigma_Y^{1/2},
v=\Sigma_Y^{-1/2}(\m-\m_Y),
$$
we obtain
\begin{equation} \label{e}
\begin{array}{l}
H^{\times}(Y\|\Nor_{(\m,\Sigma)})= \frac{1}{2}
\left(\tr(\Sigma^{-1}\Sigma_Y)+\|\m-\m_Y\|^2_{\Sigma}+\ln \det\Sigma+N\ln(2\pi)\right) \\[1ex]
=\frac{1}{2}\big(\tr(\Sigma^{-1}\Sigma_Y)+(\m-\m_Y)^T\Sigma^{-1}(\m-\m_Y)-\ln \det \Sigma^{-1}\Sigma_Y\\[0.5ex]
\phantom{=\frac{1}{2}\big(}+\ln \det \Sigma_Y+N\ln(2\pi)\big) \\[1ex]
=\frac{1}{2}\left(\tr(A)+v^TAv-\ln \det A+\ln \det \Sigma_Y+N\ln(2\pi)\right). 
\end{array}
\end{equation}
Observe that $A$ is then a symmetric positive matrix, and that given a symmetric positive matrix $A$ we can uniquely determine $\Sigma$
by the formula 
\begin{equation} \label{wyzna}
\Sigma=\Sigma_Y^{1/2}A^{-1}\Sigma_Y^{1/2}.
\end{equation}
Thus finding minimum of \eqref{e} reduces to finding a symmetric
positive matrix $A$ which minimize the value of
\begin{equation} \label{now}
\tr(A)+v^TAv-\ln \det A.
\end{equation}
Let us first consider $A \in S_{\lambda_1,\ldots,\lambda_N}$, 
where $0 < \lambda_1 \leq \ldots \leq \lambda_N$ are fixed. Our
aim is to minimize
$$
v^TAv=\tr(v^TAv)=\tr(A \cdot (vv^T)).
$$
We fix an orthonormal base such that $v/\|v\|$ is its first element,
and then by applying von Neumann trace formula we obtain that 
the above minimizes when $v$ is the eigenvector of $A$
corresponding to $\lambda_1$, and thus the minimum equals
$$
\lambda_1 \|v\|^2.
$$
Consequently we arrive at the minimization problem
$$
\lambda_1 (1+\|v\|^2)+\sum_{i>1}\lambda_i-\sum_i \ln \lambda_i. 
$$
Now one can easily verify that the minimum of the above
is realized for
$$
\lambda_1=1/(1+\|v\|^2), \lambda_i=1 \for i >1.
$$
and then \eqref{now} equals
$$
N+\ln(1+\|\m-\m_Y\|_{\Sigma_Y}^2),
$$ 
while the formula for $A$ minimizing it is given by
$$
A=\I-\frac{vv^T}{1+\|v\|^2}.
$$

Consequently then the value of \eqref{e} is
$$
\frac{1}{2}\left(\ln(1+\|\m-\m_Y\|^2_{\Sigma_Y})+\ln|\Sigma_Y|+N\ln(2\pi e)\right). 
$$
and is attained by \eqref{wyzna} for
$$
\Sigma=\Sigma_Y^{1/2}(I-\frac{\Sigma_Y^{-1/2}(\m-\m_Y)(\m-\m_Y)^T\Sigma_Y^{-1/2}}{1+\|\m-\m_Y\|_{\Sigma_Y}^2} )^{-1}\Sigma_Y^{1/2}
$$
$$
=\Sigma_Y(\Sigma_Y-\frac{(\m-\m_Y)(\m-\m_Y)^T}{1+\|\m-\m_Y\|_{\Sigma_Y}^2} )^{-1}\Sigma_Y.
$$
\end{proof}



By the above theorem we get the formula for $H^{\times}(Y\|\G_{\m})$.

\begin{corollary}\label{the:cs_4}
Let $\m \in \R^N$ be fixed. Then 
$$
H^{\times}(Y\|\G_{\m}) = \frac{1}{2}\left(\ln(1+\|\m-\m_Y\|^2_{\Sigma_Y})+\ln|\Sigma_Y|+N\ln(2\pi e)\right),
$$
and is attained for the density $\Nor(\m,\Sigma)$, where
$
\Sigma=\Sigma_Y(\Sigma_Y-\frac{(\m-\m_Y)(\m-\m_Y)^T}{1+\|\m-\m_Y\|_{\Sigma_Y}^2} )^{-1}\Sigma_Y
$.
\end{corollary}

%

\begin{table}[!h]\centering
\def\arraystretch{1.7}
\begin{tabular}{||l|l||} \hline \hline
$\F$  & $M^{\times}(Y\|\F)$   \\[0.5ex] 
\hline \hline

$\G$ & $0$
\\ \hline

$\G_{\m}$ & $\frac{1}{2}\left(\ln(1+\|\m-\m_Y\|^2_{\Sigma_Y})+\ln|\Sigma_Y|+N\ln(2\pi e)\right)$
\\ \hline

$\G_{sI}$ & $\frac{N}{2}\ln \left(\tr \frac{\Sigma_Y}{N} \right)-\frac{1}{2}\ln \det \Sigma_Y$
\\ \hline

$\G_{\m,sI}$  & $ \frac{N}{2}\ln \left( \frac{\tr(\Sigma_Y)+\|\m-\m_Y\|^2}{N} \right) -  \frac{1}{2}\ln \left( \det \Sigma_Y \right) $
\\ \hline

$\G_{\diag}$  & $ \frac{N}{2}\ln \left( \frac{\tr(\Sigma_Y) }{N} \right) -\frac{1}{2}\ln \det \Sigma_Y$
\\ \hline

$\G_{\m,\diag}$ & $\frac{1}{2}\sum \limits_{i=1}^N \ln \big((\Sigma_Y)_{ii}+|\m_i-(\m_Y)_i|^2\big)-\frac{1}{2}\ln \left( \det \Sigma_Y \right) $
\\ \hline
\hline
\end{tabular}
\caption{Table of $M^{\times}(Y\|\F)$ with respect to Gaussian subfamilies.}
\label{tab1:cec}
\end{table}

At the end of this article we illustrate our results on the data generated from the classical Lena picture. 

\begin{example}\label{ex:2}
Let us consider the classical Lena picture from The USC-SIPI Image Database (\url{http://sipi.usc.edu/database/}).
First, we interpret photo as a dataset, as in the JPG compression. We do this by dividing it into 8 by 8 pixels, where each pixel is described (in RGB) by using 3 parameters. Consequently each of the pieces is represented as a vector from $\mathbb{R}^{192}$. By this operation we obtain dataset $Y$ from $ \mathbb{R}^{192} $. 
In Table \ref{tab:jpg} we present values of $\M(Y\|\G_{\m,\F})$ for various $\F$. Moreover we consider the case when the origin is at $\m = \left[\frac{1}{2},\ldots,\frac{1}{2} \right] $, similarly like in JPG format.

\begin{table}[ht]
\def\arraystretch{1.5}
  \centering
\begin{tabular}{ || l |  c | c | c || }
\hline\hline
 &  $\M(Y\|\G_{\m})$ & $\M(Y\|\G_{\m,\diag})$  & $\M(Y\|\G_{\m,sI})$ \\
\hline
\hline
$\m = \m_{Y}$  & 0  & 672.9  & 679.168 \\
$\m = \left[\frac{1}{2},\ldots,\frac{1}{2} \right] $ & 0.520298 & 717.217  & 704.124  \\
$\m = \left[0,\ldots,0 \right] $ & 2.73125 & 883.227  & 1031  \\
\hline
\hline
\end{tabular}
  \caption{Values of $\M(Y\|\F)$ for different $\F$ and data from Example \ref{ex:2}.}
  \label{tab:jpg}
\end{table}
As a conclusion we see that the change of the origin from the mean
to the point $[\frac{1}{2},\ldots,\frac{1}{2}]$ does not ``cost'' us
much, compared to the case when we restrict to the class of 
$\G_{\diag}$. Moreover we see that the fixed center in $ \left[\frac{1}{2},\ldots,\frac{1}{2} \right]$ is better then $ \left[0,\ldots,0 \right] $ which shows that JPG approach is reasonable.
\end{example}

\section{Conclusion}

In this paper we present the method of determining the optimal coordinate systems for the dataset $Y \in \R^N$ which meets an commonly encountered conditions. We interpret the optimal Gaussian density (in the sense of cross--entropy) as an optimal transformation of a data. Consequently we obtain the measure of the optimality of the given coordinates system (represented by Gaussian family $\F$)
$$
\M(Y\|\F) = \inf_{f \in \F} H^{\times}(Y\| f)-H^{\times}(Y\|\Nor_Y).
$$

We obtain estimations in various subclasses of normal densities. The results we present in Tab. \ref{tab1:cec}

\bibliographystyle{spmpsci}      

\begin{thebibliography}{10}

\bibitem{Bo-Gr}
I.~Borg and P.~J. Groenen, {\em Modern multidimensional scaling: Theory and
  applications}.
\newblock Springer Verlag, 2005.

\bibitem{DM-JR}
R.~De~Maesschalck, D.~Jouan-Rimbaud, and D.~L. Massart, ``The mahalanobis
  distance,'' {\em Chemometrics and Intelligent Laboratory Systems}, vol.~50,
  no.~1, pp.~1--18, 2000.

\bibitem{Ha-Ka}
J.~Han, M.~Kamber, and J.~Pei, {\em Data mining: concepts and techniques}.
\newblock Morgan kaufmann, 2006.

\bibitem{Ja-Sa}
T.~Jayalakshmi and A.~Santhakumaran, ``Statistical normalization and back
  propagation for classification,'' {\em International Journal of Computer
  Theory and Engineering}, vol.~3, no.~1, pp.~1793--8201, 2011.

\bibitem{Ra-Ma}
T.~Raykov and G.~A. Marcoulides, {\em An introduction to applied multivariate
  analysis}.
\newblock Psychology Press, 2008.

\bibitem{Re}
A.~C. Rencher and W.~F. Christensen, {\em Methods of multivariate analysis},
  vol.~709.
\newblock Wiley, 2012.

\bibitem{Ti}
N.~H. Timm, {\em Applied multivariate analysis}.
\newblock Springer Verlag, 2002.

\bibitem{Kr-Ka}
P.~R. Krishnaiah and L.~N. Kanal, ``Classification, pattern recognition, and
  reduction of dimensionality, volume 2 of handbook of statistics,'' {\em
  North4Holland Amsterdam}, 1982.

\bibitem{ma_dis_a_3}
M.~P. McAssey, ``An empirical goodness-of-fit test for multivariate
  distributions,'' {\em Journal of Applied Statistics}, no.~ahead-of-print,
  pp.~1--12, 2013.

\bibitem{ma_dis_a_4}
E.~J. Bedrick, ``Graphical modelling and the mahalanobis distance,'' {\em
  Journal of Applied Statistics}, vol.~32, no.~9, pp.~959--967, 2005.

\bibitem{ma_dis_a_5}
W.~Krzanowski, ``Non-parametric estimation of distance between groups,'' {\em
  Journal of Applied Statistics}, vol.~30, no.~7, pp.~743--750, 2003.

\bibitem{ma_dis_a_1}
A.~F. Mitchell and W.~J. Krzanowski, ``The mahalanobis distance and elliptic
  distributions,'' {\em Biometrika}, vol.~72, no.~2, pp.~464--467, 1985.

\bibitem{ma_dis_a_2}
H.~Holgersson and P.~S. Karlsson, ``Three estimators of the mahalanobis
  distance in high-dimensional data,'' {\em Journal of Applied Statistics},
  vol.~39, no.~12, pp.~2713--2720, 2012.

\bibitem{Ma}
P.~C. Mahalanobis, ``On the generalized distance in statistics,'' in {\em
  Proceedings of the national institute of sciences of India}, vol.~2,
  pp.~49--55, New Delhi, 1936.

\bibitem{jolliffe2005principal}
I.~Jolliffe, {\em Principal component analysis}.
\newblock Wiley Online Library, 2005.

\bibitem{pca_a_1}
I.~T. Jolliffe, ``Rotation of principal components: choice of normalization
  constraints,'' {\em Journal of Applied Statistics}, vol.~22, no.~1,
  pp.~29--35, 1995.

\bibitem{pca_a_2}
P.~Pack, I.~Jolliffe, and B.~Morgan, ``Influential observations in principal
  component analysis: A case study,'' {\em Journal of Applied Statistics},
  vol.~15, no.~1, pp.~39--52, 1988.

\bibitem{Le-Ca}
E.~L. Lehmann and G.~Casella, {\em Theory of point estimation}, vol.~31.
\newblock Springer, 1998.

\bibitem{Bo}
A.~Van~den Bos, {\em Parameter estimation for scientists and engineers}.
\newblock Wiley-Interscience, 2007.

\bibitem{mle_a_3}
R.~Cheng and N.~Amin, ``Maximum likelihood estimation of parameters in the
  inverse gaussian distribution, with unknown origin,'' {\em Technometrics},
  vol.~23, no.~3, pp.~257--263, 1981.

\bibitem{Ni-No}
F.~Nielsen and R.~Nock, ``Clustering multivariate normal distributions,'' in
  {\em Emerging Trends in Visual Computing}, pp.~164--174, Springer, 2009.

\bibitem{Pe}
W.~D. Penny, ``Kullback-liebler divergences of normal, gamma, dirichlet and
  wishart densities,'' {\em Wellcome Department of Cognitive Neurology}, 2001.

\bibitem{Gr}
R.~D. Grigorieff, ``A note on von neumann's trace inequality,'' {\em Math.
  Nachr}, vol.~151, pp.~327--328, 1991.

\bibitem{Mi}
L.~Mirsky, ``A trace inequality of john von neumann,'' {\em Monatshefte f{\"u}r
  Mathematik}, vol.~79, no.~4, pp.~303--306, 1975.

\end{thebibliography}


\end{document}